\documentclass[journal]{IEEEtran}
\RequirePackage{filecontents}

\usepackage{epsfig,amsmath,amssymb,epsf,amsthm,scalefnt,multirow,subfig}
\usepackage{xcolor}
\usepackage{float}
\usepackage{cite}
\usepackage{psfrag}\usepackage{graphics,amsmath,amssymb,graphicx,amsthm,scalefnt,multirow,dsfont,subfig}
\usepackage{algorithm}
\usepackage{algpseudocode}
\usepackage{verbatim}
\usepackage{tabularx,booktabs}
\usepackage{multicol}
\usepackage{stfloats}
\usepackage{bm}
\usepackage{kotex} 
\usepackage{booktabs}
\usepackage{xcolor}
\usepackage{epsfig}
\usepackage{epstopdf}
\usepackage{acronym}
\graphicspath{{fig/}}
\usepackage{mathtools}
\usepackage{amsthm}
\usepackage{bm}
\usepackage{comment}
 \usepackage{subcaption}
 \usepackage{setspace}
  \usepackage{graphicx}
  \usepackage{soul}
  \usepackage[absolute,overlay]{textpos}

\usepackage{ragged2e}
\newtheorem{theorem}{Theorem}
\newtheorem{lemma}{Lemma}
\newtheorem*{lemma*}{Lemma}

\def\cA{{\mathcal{A}}}

\def\b0{{\pmb{0}}}

\def\ba{{\mathbf{a}}} \def\bb{{\mathbf{b}}}

   \def\bp{{\mathbf{p}}}
 \def\br{{\mathbf{r}}} \def\bs{{\mathbf{s}}}

\makeatletter
\def\@IEEEsectpunct{.\ \,}
\def\paragraph{\@startsection{paragraph}{4}{\z@}{1.2ex plus 1.1ex minus 0.5ex}%
{0ex}{\normalfont\normalsize\bfseries}}
\makeatother

\theoremstyle{remark}

\begin{document}

\title{Analysis of Beam Misalignment Effect \\ in Inter-Satellite FSO Links}

\author{\IEEEauthorblockN{Minje Kim, Hongjae Nam, Beomsoo Ko, Hyeongjun Park, Hwanjin Kim, Dong-Hyun Jung, and Junil Choi}
\thanks{M. Kim, B. Ko, H. Park, and J. Choi are with the School of Electrical Engineering, Korea Advanced Institute of Science and Technology, Daejeon 34141, South Korea (e-mail: mjkim97@kaist.ac.kr; kobs0318@kaist.ac.kr; mika0303@kaist.ac.kr; junil@kaist.ac.kr)}
\thanks{H. Nam is with the Elmore Family School of Electrical
and Computer Engineering, Purdue University, West Lafayette, IN 47907 USA
(e-mail: nam86@purdue.edu).}
\thanks{H. Kim is with the School of Electronics Engineering, Kyungpook National
University, Daegu 41566, South Korea (e-mail: hwanjin@knu.ac.kr).}
\thanks{D-H. Jung is with the School of Electronic Engineering, Soongsil
University, Seoul, 06978, South Korea (e-mail: dhjung@ssu.ac.kr).}
}

\maketitle

\begin{abstract}
Free-space optical (FSO) communication has emerged as a promising technology for inter-satellite links (ISLs) due to its high data rate, low power consumption, and reduced interference. However, the performance of inter-satellite FSO systems is highly sensitive to beam misalignment.
While pointing-ahead angle (PAA) compensation is commonly employed, the effectiveness of PAA compensation depends on precise orbital knowledge and advanced alignment hardware, which are not always feasible in practice. To address this challenge, this paper investigates the impact of beam misalignment on inter-satellite FSO communication. We derive a closed-form expression for the cumulative distribution function (CDF) of the FSO channel under the joint jitter and misalignment-induced pointing error, and introduce a truncated CDF formulation with a bisection algorithm to efficiently compute outage probabilities
with guaranteed convergence and minimal computational overhead. To make the analysis more practical, we quantify displacement based on orbital dynamics. Numerical results demonstrate that the proposed model closely matches Monte Carlo simulations, making the proposed model highly useful to design inter-satellite FSO systems in practice.
\end{abstract}


\section{Introduction}
Satellite communication systems have gained interest in recent years due to their unique ability to provide global coverage, particularly benefiting remote and underserved regions where terrestrial infrastructure is limited or nonexistent~\cite{cioni2018satellite}. A key enabler of this capability is the deployment of satellite constellations in low Earth orbit (LEO), where a large number of coordinated satellites operate collectively to improve network performance. Such constellations enhance system efficiency by enabling low-latency communication, wide-area coverage, and high-resolution imaging from various orbital altitudes~\cite{giordani2020non}. Effective operation of constellations requires coordination among satellites~\cite{jung2023satellite}, highlighting the importance of dedicated inter-satellite communication technologies.

In this context, the development of inter-satellite links (ISLs) has become essential for enabling direct communication between satellites without dependence on ground stations, and supporting advanced operational concepts such as formation flying~\cite{radhakrishnan2016_ISLintro2, chen2021_ISLintro1}. By forming a space-based mesh network, ISLs facilitate flexible routing, substantially reduce end-to-end latency, and extend robust communication coverage, particularly in remote areas~\cite{kodheli2020satellite}. Such capabilities are critical for LEO constellations, where satellite cooperation through ISLs can significantly enhance overall network efficiency and performance~\cite{kaushal:2016vacuum}. Furthermore, ISLs contribute to increased system robustness by effectively offloading traffic and maintaining continuous, reliable service between the space segment and ground users~\cite{chen2024survey}.

While both radio frequency (RF) and free-space optical (FSO) communication technologies can be considered for ISLs, FSO communication has recently gained significant attention owing to its ability to support higher data rates, lower power consumption, and reduced interference. In particular, the advantages of FSO are prominent in long-distance scenarios, highlighting its suitability for high-capacity ISLs~~\cite{trichili2020roadmap}.
However, the performance of ISLs is highly sensitive to beam pointing error, as FSO links are vulnerable due to narrow beamwidth and high directivity, which can degrade overall link performance. Even small deviations in the line-of-sight (LoS) can lead to significant power loss at the receiver~\cite{chen:1989pointingerror}.

Therefore, accurate analysis of pointing error is essential for evaluating the performance of inter-satellite FSO systems.
The pointing error in FSO systems consists of two main sources, misalignment and jitter, where \textit{misalignment} refers to a static displacement of the beam center, typically caused by orbital uncertainty or limited attitude precision, while \textit{jitter} denotes fast fluctuations around the beam axis, often induced by mechanical vibration or actuator noise.

In inter-satellite FSO links, pointing-ahead angle (PAA) compensation is commonly employed to mitigate misalignment, where a transmit satellite steers its beam toward the predicted future location of the receiving satellite~\cite{calvo2019paa}. The effectiveness of PAA compensation, however, highly depends on the availability of precise orbital knowledge and accurate control, which may not be feasible in practice~\cite{toyoshima:2005trends}. Although high-end commercial satellites, e.g., Starlink, can incorporate dedicated alignment hardware, such systems are very costly for resource-constrained platforms such as CubeSats~\cite{yoon2017pointing ,chaudhry:2021starlink}. Instead, these smaller satellites typically rely on body-pointing via their attitude control systems, which results in slower response and beam misalignment~\cite{cahoy:2019click, long:2018pointing}. Therefore, evaluating the performance of inter-satellite FSO links under imperfect beam alignment is essential, particularly in practical environments where ideal compensation cannot be guaranteed.

Recognizing the limitations of current compensation techniques, we evaluate the inter-satellite FSO communication system under beam misalignment scenarios. 
In this paper, we provide an inter-satellite FSO channel modeling under the misalignment condition and derive its exact cumulative density function (CDF) formulation. 
Since the exact CDF involves an infinite series, a truncation algorithm is introduced to enable efficient and practical computation of outage probabilities. 
The analysis incorporates the influence of satellite spatial configuration on beam misalignment, providing a comprehensive assessment of system performance.

\subsection{Related Work}
Various statistical distributions have been proposed to characterize pointing error in FSO communication systems~\cite{farid:2007,gappmair2011,yang:2014nonzero,alquwaiee2016asymptotic}. 
The Rayleigh distribution captures the case where the pointing error results purely from jitter with zero-mean and equal variance in both directions~\cite{farid:2007}. 
The Hoyt distribution extends the analysis by allowing different jitter variances along the horizontal and vertical axes, while still assuming no misalignment~\cite{gappmair2011}. 
The Rician distribution incorporates the effect of nonzero-mean displacement, modeling misalignment under equal jitter variances~\cite{yang:2014nonzero}. 
The Beckmann distribution provides the most general model, capturing both nonzero mean displacement and unequal jitter variances~\cite{alquwaiee2016asymptotic}.

{In terrestrial FSO systems, these statistical models have been widely applied to analyze the combined effects of beam misalignment, atmospheric turbulence, and mechanical jitter.} 
Several studies have investigated diverse FSO scenarios, including hybrid FSO/RF relaying systems~\cite{nguyen2022rateadaptation,zedini2016dualhop,soleimani2016generalized,lei2018secrecy}, 
aerial platforms such as unmanned aerial vehicles (UAVs) and high altitude platforms (HAPs)~\cite{dabiri2018uav,wang2021hovering,safi2020hap}, 
and multi-aperture MIMO configurations~\cite{bhatnagar2016mimo}. 
However, these previous studies have primarily focused on terrestrial environments, where a combination of beam pointing error and atmospheric turbulence effects leads to analytical intractability, 
necessitating the use of approximation or asymptotic methods in performance analysis.

In contrast, inter-satellite FSO communication operates in vacuum conditions, effectively eliminating turbulence-induced fading, which allows a more isolated and focused characterization of pointing error effects. Several studies have investigated the impact of pointing errors in inter-satellite FSO communication systems~\cite{nelson1994experimental,arnon2005ppointingerror,polishuk2004optimization,liang2023free,tawfik2021performance,song2017impact,zhu2024average}.
The pointing error distribution caused by hardware jitter was experimentally characterized and shown to closely approximate a Gaussian distribution \cite{nelson1994experimental}. Link margin performance of inter-satellite FSO systems under pointing errors was evaluated, analyzing the resulting bit error rate (BER) for given link margins \cite{arnon2005ppointingerror}. Optimal transmit power to achieve target BER was computed by taking jitter-induced pointing errors into account \cite{polishuk2004optimization}. More recent studies examined the required transmitted power for end-to-end satellite communication chains involving multiple ISLs affected by pointing errors \cite{liang2023free}, and the performance of long-distance ISLs between LEO and geostationary Earth orbit (GEO) satellites \cite{tawfik2021performance}. 
{Some studies have modeled pointing errors by integrating the spatial intensity distribution of the optical Gaussian beam with the finite receiver aperture size~\cite{song2017impact, zhu2024average}. 
These physical-layer approaches enable a more detailed characterization of pointing error.}

{While the above studies addressed pointing errors either at the system level through link budget and BER analysis~\cite{nelson1994experimental,arnon2005ppointingerror,polishuk2004optimization,liang2023free,tawfik2021performance}, 
or at the physical layer via received power modeling using beam intensity profiles~\cite{song2017impact,zhu2024average}, 
they focused on stochastic jitter and assumed perfect beam alignment. 
As a result, beam misalignment arising from deterministic factors such as orbital dynamics or limited tracking control has not been explicitly incorporated into existing ISL performance models.}

\subsection{Challenges and Contributions}
\paragraph*{Challenges}
Although pointing errors in inter-satellite FSO communication systems have been studied, several critical challenges related to realistic pointing error modeling and dynamic constellation effects remain insufficiently addressed. The main challenges can be summarized as follows.
    \begin{itemize}

\item \textbf{Limitations in existing pointing error models:} Prior research focused mainly on stochastic jitter effects while assuming perfect compensation of static misalignment, corresponding to ideal tracking conditions~\cite{song2017impact,zhu2024average}. Such assumptions are inadequate for small satellites like CubeSats, where limited attitude control and orbital dynamics induce systematic beam misalignment that should be taken into account. Comprehensive performance evaluation demands joint modeling of both random jitter and misalignment components.

\item \textbf{Simplified satellite geometry and dynamics:} Existing analyses utilized fixed inter-satellite distances derived from constellation configurations and treated links independently without accounting for relative satellite motion or dynamic changes in geometry over time~\cite{arnon2005ppointingerror,polishuk2004optimization}. This simplification disregards dynamic beam misalignment effects and constellation-scale interactions driven by coordinated satellite trajectories. Practical performance assessment requires incorporating time-varying orbital dynamics alongside constellation topology.
\end{itemize}

\paragraph*{Contributions} 
To overcome the challenges, this paper presents novel modeling techniques and performance evaluation algorithms that capture the complexities of inter-satellite FSO links. The key contributions are outlined as follows.
\begin{itemize}
    \item \textbf{Statistical modeling of pointing error in inter-satellite FSO links}: 
Previous studies have considered non-zero misalignment only in terrestrial/aerial FSO systems~\cite{yang:2014nonzero, alquwaiee2016asymptotic,nguyen2022rateadaptation,zedini2016dualhop,soleimani2016generalized,lei2018secrecy,dabiri2018uav,wang2021hovering,safi2020hap,bhatnagar2016mimo}. In this work, we present a unified model of pointing errors that incorporates both jitter and misalignment, and establish an exact closed-form CDF expression for the inter-satellite FSO link channel. To make the analysis computationally tractable, we further introduce a truncation-based algorithm that enables efficient evaluation of outage probabilities. The proposed framework combines analytical rigor with computational practicality and can be directly applied to ISL performance assessment and system design.

    \item \textbf{Analytical quantification of misalignment induced by orbital dynamics:} 
    {We develop an analytical approach to quantify misalignment arising from satellite orbital motion. While terrestrial FSO studies have generally not quantified motion-induced effects, our method incorporates orbital dynamics and the relative orientation of orbital planes into a receiver-centric spherical coordinate framework. A bisection-based algorithm is employed to determine signal arrival time and estimate receiver displacement with high computational efficiency.}
    
    \item \textbf{Constellation-based performance evaluation:}  
    To assess practical relevance, we analyze the inter-satellite FSO performance under realistic Iridium and Starlink constellation scenarios. For each constellation, we compute the ISL distance and corresponding displacement, and evaluate the resulting outage probability. Simulation results verify that our proposed method matches the performance of Monte Carlo simulations while offering significant computational efficiency.
\end{itemize} 

The remainder of this paper is organized as follows. Section~\ref{sec2} describes the inter-satellite FSO system model, including the Gaussian beam and pointing error characterization. Section~\ref{sec3} derives closed-form expressions for the PDF and CDF of the pointing error channel, along with a practical truncation approximation for the CDF. Section~\ref{sec4} presents a numerical method to compute misalignment displacement due to satellite motion, using a receiver-centric frame and a bisection algorithm for signal arrival time. Section~\ref{sec5} evaluates the pointing error induced by joint jitter and beam misalignment via simulations of Iridium and Starlink constellations. Finally, conclusions are given in Section~\ref{sec6}.

\begin{figure}
    \centering
    \includegraphics[width=1\linewidth]{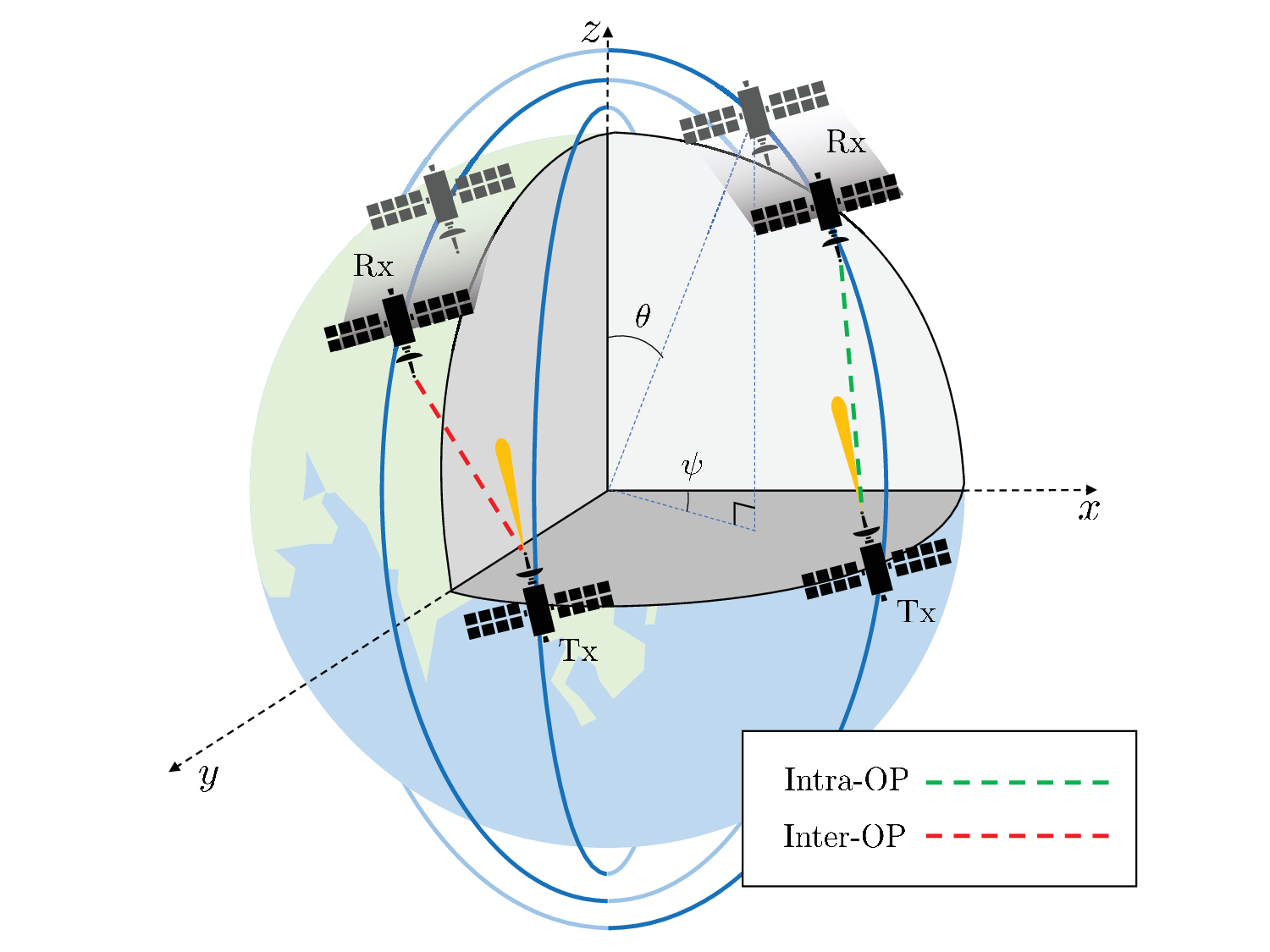}
    \caption{Inter-satellite FSO communication system with transmit satellite and receiver satellite. Both intra-orbital plane (intra-OP) and inter-orbital plane (inter-OP) links are shown in the figure.}
    \label{fig:sys_model}
\end{figure}

\section{System Model} \label{sec2}

We consider an inter-satellite communication system where two satellites are each equipped with an FSO transceiver, as illustrated in Fig.~\ref{fig:sys_model}. {ISLs are typically classified into two types, intra-orbital plane (intra-OP) and inter-orbital plane (inter-OP) links, depending on whether the satellites are in the same or different orbital planes (OPs)}.
{In the system}, the vacuum environment of space eliminates atmospheric turbulence that is a major channel fluctuation factor in terrestrial FSO communication systems~\cite{sharma:2013vacuum,kaushal:2016vacuum}. As a result, the inter-satellite FSO channel is only affected by pointing error, which consists of jitter and beam misalignment.

In FSO systems, the transmitted beam is typically modeled as a Gaussian beam due to its analytical tractability and accuracy over long distances. With the Gaussian beam model,  the received intensity at a distance \( \ell \) from the transmitter is given by~\cite{saleh:2008}
\begin{equation}
    I_{\text{beam}}(\boldsymbol{\rho}; \ell) = \frac{2}{\pi \omega_\ell^{2}} \exp\left(-\frac{2\lVert\boldsymbol{\rho}\rVert^{2}}{\omega_\ell^{2}}\right),
\end{equation}
where \( \boldsymbol{\rho} \) denotes the radial distance from the beam axis, and \( \omega_\ell \) is the beam radius at distance \( \ell \), defined as the point where the intensity falls to \( 1/e^2 \) of its maximum. The graphical explanation of the Gaussian beam model is shown in Fig.~\ref{fig:Gaussian beam}.
The beam profile is fully determined by the beam waist \( \omega_0 \) at \( \ell = 0 \), where the beam radius is minimized. The relationship between \( \omega_0 \) and \( \omega_\ell \) is given by~\cite{alda2003laser}
\begin{equation} \label{eq:beam radius}
    \omega_\ell = \omega_0 \sqrt{1 + \left( \frac{\lambda \ell}{\pi \omega_0^2} \right)^2 },
\end{equation}
where \( \lambda \) denotes the optical wavelength.

\begin{figure}
    \centering
    \includegraphics[width=1\linewidth]{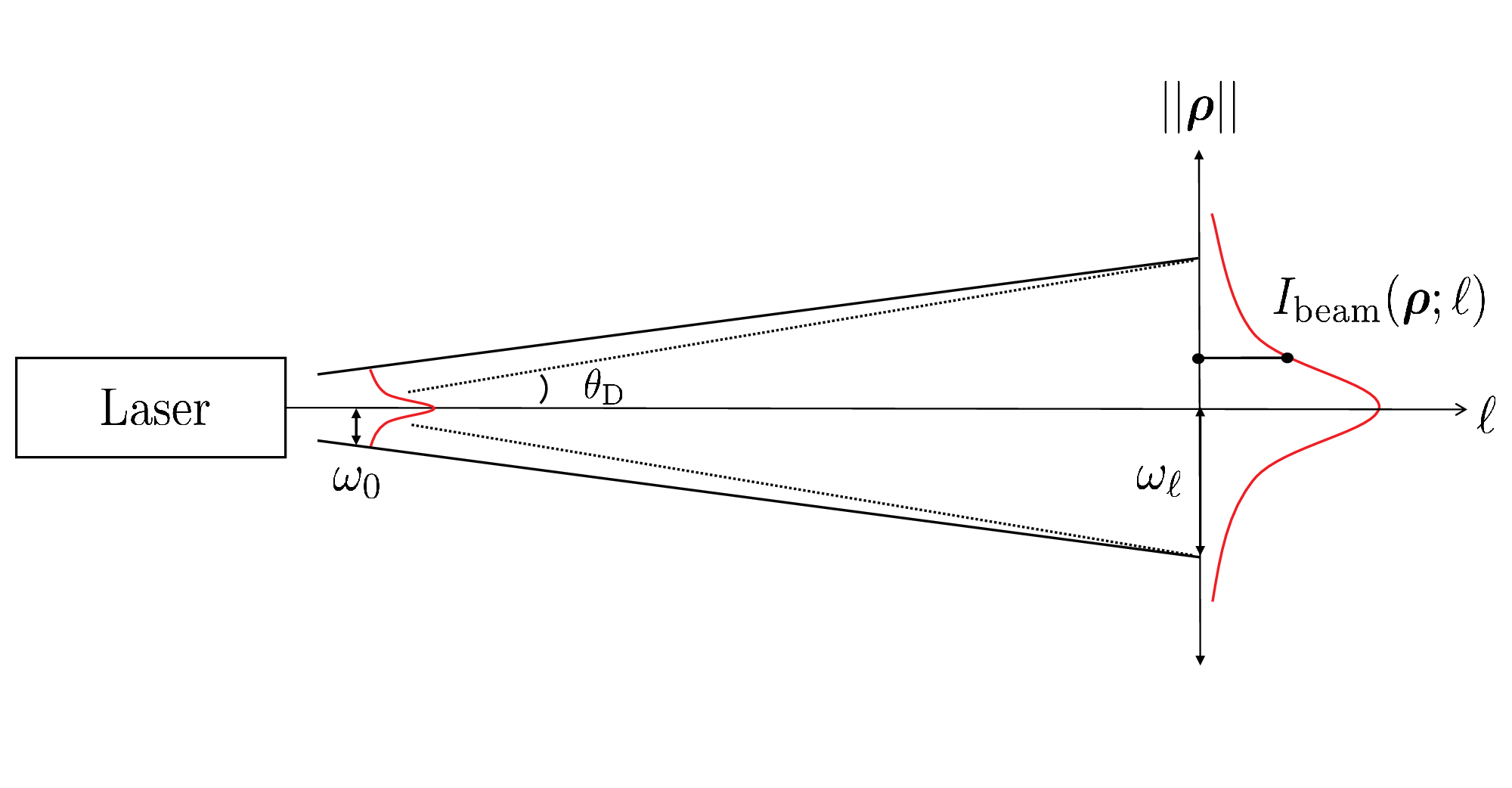}
    \caption{Beam waist at the transmitter and divergence of the beam along the propagation axis with Gaussian beam model.}
    \label{fig:Gaussian beam}
\end{figure} 
With the Gaussian beam model, the collected power at the receiver in the presence of a pointing error $\br$ is given by
\begin{equation} \label{eq: beamgain}
    h(\br;\ell) = \int_\cA I_{\rm beam}(\boldsymbol{\rho} - \br; \ell) d\boldsymbol{\rho},
\end{equation}
where $\cA$ denotes the antenna area. Note that $h(\cdot)$ represents the channel gain in the context of inter-satellite FSO communication.
Considering an antenna area $\cA$ with a circular aperture of radius $a$, and defining $r = \|\br\|_2$, we can approximate the above integration as~\cite{farid:2007}
\begin{equation} 
    h(r;\ell) \approx A_0 \exp \left(- \frac{2r^2}{\omega_{\ell_{\rm eq}}^2} \right), \label{ep:h_approx}
\end{equation}
where 
\begin{equation*}
    A_0 = [\mathrm{erf}(v)]^2,~ 
    \omega_{\ell_{\rm eq}}^2 = \omega_\ell^2 \frac{\sqrt{\pi} \, \mathrm{erf}(v)}{2v \exp(-v^2)},
    v = \frac{\sqrt{\pi} a}{\sqrt{2} \omega_\ell}.
\end{equation*}
The function $\mathrm{erf}(\cdot)$ refers to the Gauss error function.
Then, for a given FSO channel gain between the satellites, the signal-to-noise ratio (SNR) of the inter-satellite FSO link is given by
\begin{equation}
    \mathrm{SNR}(h) = \frac{h^2 R^2 P_t^2}{\sigma_n^2},
\end{equation}
where $R$ is the photoelectric responsivity of the receiver, $P_t$ represents the transmit power, and $\sigma_n^2$ is the noise variance.

To evaluate the performance of the inter-satellite FSO channel, we adopt the outage probability as the performance metric. The outage probability at a target transmission rate $R_0$ can be written in terms of the channel gain $h$ as
{\begin{equation}
    P_{\mathrm{out}}(R_0) = \mathrm{Prob} \left( \mathcal{C}(\mathrm{SNR}(h)) < R_0 \right),
\end{equation}
where $\mathcal{C}(\cdot) = \log_2(1+\cdot)$ denotes the instantaneous capacity.
The above expression can be equivalently rewritten in terms of the channel gain $h$ as}
\begin{equation} \label{eq: outage_h}
    P_{\mathrm{out}}(R_0) = {\mathrm{Prob}} ( h < \mathrm{SNR}^{-1} \left( \mathcal{C}^{-1}(R_0) \right) ),
\end{equation}
where $\mathrm{SNR}^{-1}(\cdot)$ and $\mathcal{C}^{-1}(\cdot)$ are the inverse functions of the SNR and capacity expressions, respectively.
To compute the outage probability in \eqref{eq: outage_h}, the CDF of the channel gain $h$ is required. In Section~\ref{sec3}, we thus derive the CDF of inter-satellite FSO channel by modeling the randomness of the pointing error $r$.

\section{Pointing Error Channel Model} \label{sec3}

 \begin{figure}
    \centering
    \includegraphics[width=1\linewidth]{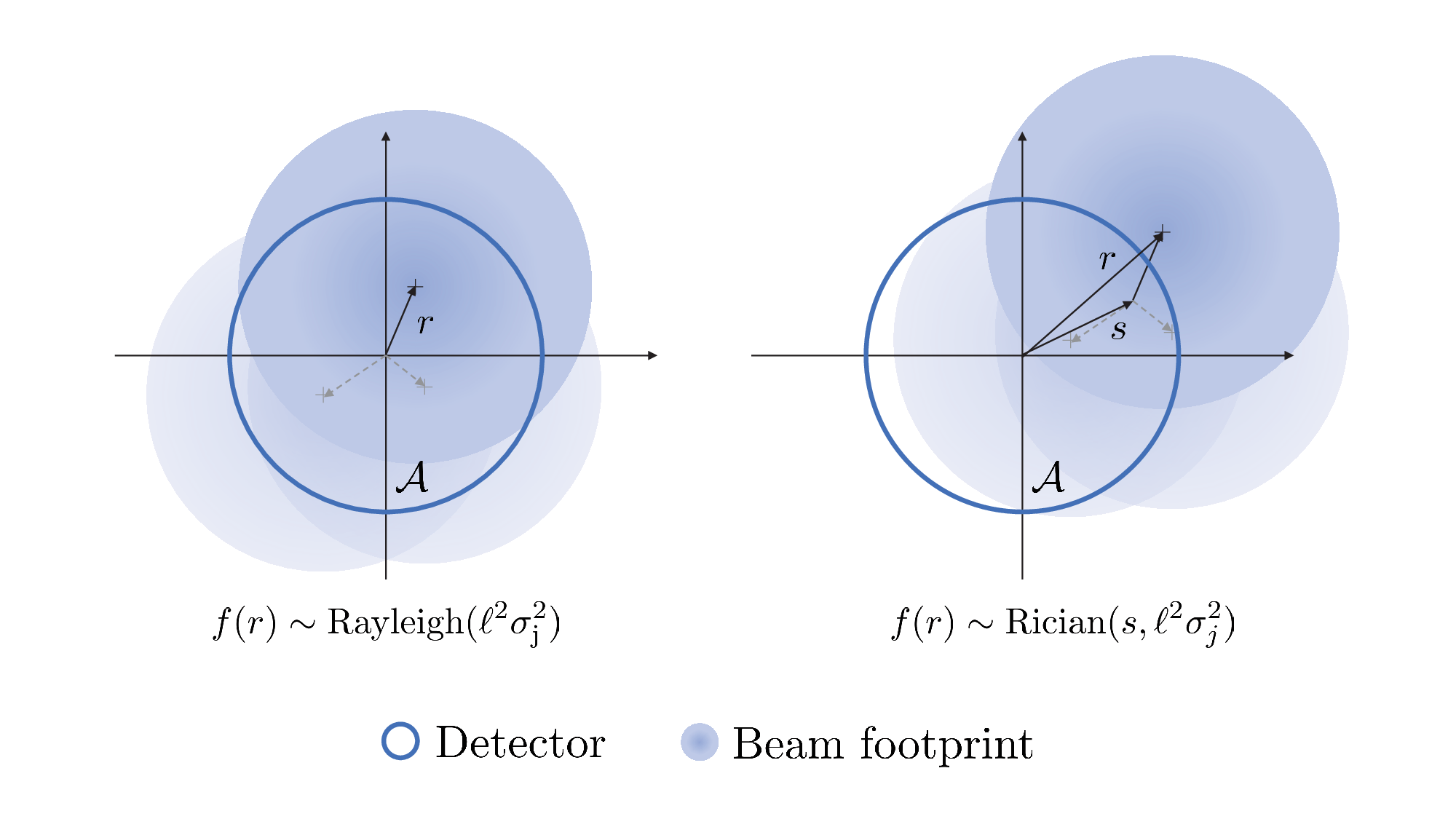}
    \caption{{Jitter-induced pointing error (left) and joint jitter and misalignment-induced pointing error (right).}}
    \label{fig:pointing error compare}
\end{figure}

To describe practical inter-satellite FSO communication systems, we consider a general pointing error expression that is jointly affected by jitter and misalignment. 
In inter-satellite FSO links, although mechanical disturbances may occur at both the transmitter and receiver satellites, fluctuations at the receiver side typically induce negligible impact, provided the optical beam remains within the field of view of the receive aperture~\cite{badas2024opto}. In contrast, transmitter-side jitter directly perturbs the beam propagation angle. Due to the high directivity and narrow divergence of inter-satellite optical beams, such angular deviations at the transmitter can result in significant lateral displacement at the receive plane.

Let \( \theta_j \) denote the angular deviation caused by transmitter jitter. Then, the corresponding jitter-induced lateral offset at a propagation distance \( \ell \) is given by
\begin{equation}
    r_j = \ell \tan(\theta_j),
\end{equation}
which simplifies under the small-angle approximation to
\begin{equation}
    r_j \approx \ell \theta_j.
\end{equation}
Assuming symmetry in the mechanical jitter dynamics of the satellite platform, the variances along horizontal and vertical axes are equal~\cite{gudimetla2007moment}.
As the jitter introduces a zero-mean random perturbation, the jitter variance at the receiver plane with propagation distance \(\ell\) becomes
\begin{equation}
    \mathrm{Var}(r_j) = \ell^2 \sigma_j^2,
\end{equation}
where \(\sigma_{j}\) denotes the angular deviation of \( \theta_j \) due to jitter.

In addition to random jitter, static misalignment may occur due to imperfect tracking or limited compensation of PAA. This results in a non-zero displacement between the beam axis and the center of the receive aperture. The combined effect of random jitter and misalignment manifests as a radial pointing error, whose distribution at distance 
\(\ell\) follows a Rician distribution~\cite{chen:1989pointingerror,arnon2005ppointingerror}
\begin{equation} \label{eq:r}
    f_r(r;\ell,\sigma_j,s) = \frac{r}{\ell^2\sigma_j^2} e^{-\frac{r^2 + s^2}{2\ell^2\sigma_j^2}} I_0\left(\frac{rs}{\ell^2\sigma_j^2}\right), \quad r \ge 0,
\end{equation}
where $s$ is the displacement due to misalignment and $I_0(\cdot)$ is the modified Bessel function of the first kind. A visual representation of the radial pointing displacement caused by jitter and its combination with static misalignment can be found in Fig.~\ref{fig:pointing error compare}, offering an intuitive understanding of the error distribution at the receive plane.

\begin{lemma}
For a given displacement value $s$, the PDF of the corresponding FSO channel gain $h$ can be obtained by combining the beam power collection model in~\eqref{ep:h_approx} with the Rician-distributed pointing error in~\eqref{eq:r}. This yields the following PDF expression
\begin{equation} \label{eq:pdf}
    f_h (h ) = \frac{\gamma^2}{A_0^{\gamma^2}}  e^{- \frac{s^2}{2\ell^2\sigma_j^2} } h^{\gamma^2 -1} I_0\left(\frac{s}{\ell^2\sigma_{j}^2} \sqrt{-\frac{\omega_{\ell_{\rm eq}}^2}{2} \ln{\frac{h}{A_0}}} \right),
\end{equation}
where the valid domain is $h \in [0, A_0]$.
\end{lemma}

\begin{proof}
    See Appendix~A for a detailed proof.
\end{proof}

To evaluate the outage probability, the CDF of $h$ is also required. Based on the PDF in \eqref{eq:pdf}, the CDF can be derived in a closed form as follows.

\begin{lemma}
Let $F_h(h)$ be the CDF of the FSO channel gain. Then, for $0 \le h \le A_0$, $F_h(h)$ is given by
\begin{align}
    F_h(h) = e^{- \frac{s^2}{2 \ell^2 \sigma_j^2}} \left \{ \sum_{n=0}^{\infty} \frac{ \left( \frac{s^2}{2 \ell^2 \sigma_j^2} \right)^n }{(n!)^2} \Gamma\left(n + 1, \gamma^2 \ln \frac{A_0}{h} \right) \right \},
\end{align}
where $\Gamma(n,x) = \int_{x}^\infty t^n e^{-t} dt$ is the upper incomplete gamma function. \end{lemma}

\begin{proof}
    See Appendix~B for a detailed proof.
\end{proof}

Since the exact CDF involves an infinite series that is computationally intractable, we consider a truncated approximation given by
\begin{equation} \label{eq:truncated CDF}
    \hat{F}_h (h) = e^{- \nu } \left\{ \sum_{n=0}^{N-1}  \frac{\nu^n}{(n !)^2} \Gamma\left(n+1, \zeta\right) \right\},
\end{equation}
where \( \nu = {s^2}/{(2\ell^2\sigma_j^2)} \) and \( \zeta = \gamma^2 \ln \left( {A_0}/{h} \right) \).
To ensure the accuracy of the approximation, we {need to} determine the minimum integer $N$ for which the contribution of the remaining terms becomes negligible.

Due to the monotonic decreasing property of the upper incomplete gamma function with respect to its lower limit \( \zeta \), it holds that
\begin{equation}
\frac{\nu^n}{(n!)^2} \Gamma(n+1, \zeta) \le \frac{\nu^n}{(n!)^2} \Gamma(n+1, 0), \quad \text{for } \zeta \ge 0.
\end{equation}
We define the right-hand side as an upper bound
\begin{equation}
U_n = \frac{\nu^n}{(n!)^2} \Gamma(n+1, 0).
\end{equation}
If \( U_n < \epsilon \) for a small threshold \( \epsilon \), then the actual \( n \)-th term is also guaranteed to be smaller than \( \epsilon \). This provides a reliable criterion for determining the truncation index \( N \).
\begin{algorithm}[t]
\caption{Bisection-Based Truncation Index Search}
\label{alg:truncation_search}
\begin{algorithmic}[1]
    \State \textbf{Input:} $\nu$, initial bound $N_{\mathrm{init}}$, threshold $\epsilon$
    \State $N_{\min} \gets 1$, $N_{\max} \gets N_{\mathrm{init}}$
    \While{$N_{\max} - N_{\min} > 1$}
        \State $N \gets \lfloor (N_{\min} + N_{\max}) / 2 \rfloor$
        \State Compute $\ln(U_N) \gets N(1 + \ln \nu - \ln N) - \frac{1}{2} \ln(2\pi N)$
        \If{$\ln(U_N) < \ln \epsilon$}
            \State $N_{\max} \gets N$
        \Else
            \State $N_{\min} \gets N$
        \EndIf
    \EndWhile
    \State \textbf{Return:} $N \gets N_{\max}$
\end{algorithmic}
\end{algorithm}

At $\zeta = 0$, the upper incomplete gamma function reduces to the complete gamma function as 
\begin{equation}
\Gamma(n+1, 0) = \Gamma(n+1) = n!,
\end{equation} so the bound can be simplified to 
\begin{equation}
U_n = {\nu^n}/{n!}.
\end{equation}
To avoid numerical instability when dealing with large $n$, we take the logarithm of both sides as 
\begin{equation}
\ln(U_n) = -\ln(n!) + n \ln \nu.
\end{equation}
Additionally, by applying the Stirling’s approximation~\cite{abramowitz1972handbook} to the factorial term, we can obtain
\begin{equation}
    \ln(U_n) \approx  n \left( 1 + \ln\nu - \ln n \right) - \frac{1}{2} \ln(2\pi n).\label{eq:lnUn_final}
\end{equation}
\noindent To determine the minimum truncation index $N$ satisfying the condition $\ln(U_N) < \ln\epsilon$, we utilize the bisection method, which enables fast and reliable convergence without evaluating factorials directly. Once the smallest $N$ is found, the truncated CDF expression in~\eqref{eq:truncated CDF} can be computed accordingly.
{The entire procedure is summarized in Algorithm~\ref{alg:truncation_search}, where the bisection search achieves convergence in $\mathcal{O}(\log_2(1/\epsilon))$ iterations~\cite{bisection_complexity}, ensuring minimal computational overhead while maintaining numerical accuracy.}

\section{Displacement Computation} \label{sec4}

\begin{figure}[t]
    \centering    \includegraphics[width=1\linewidth]{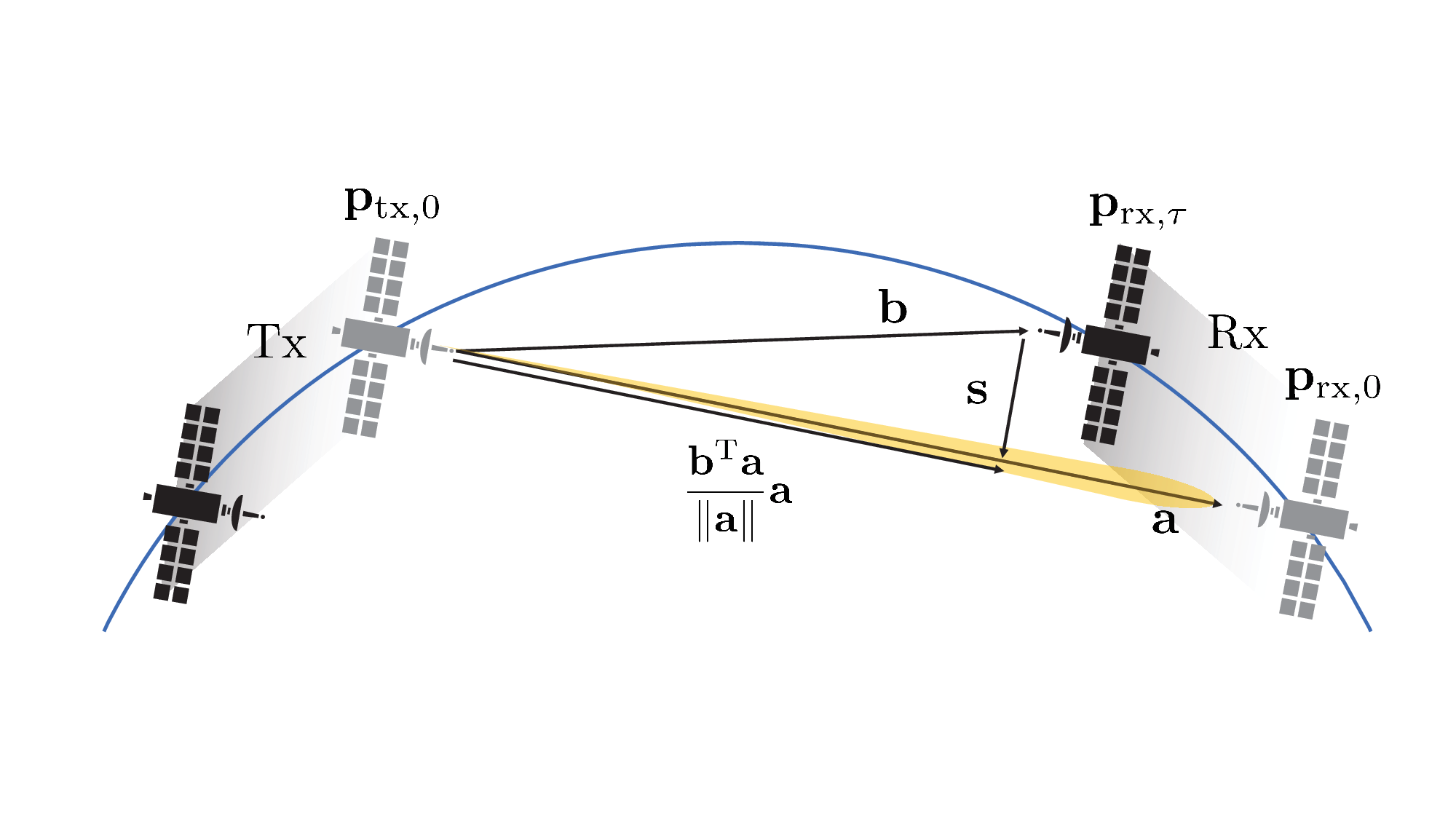}
    \caption{Misalignment due to movement of satellites.}
    \label{fig:displacement}
\end{figure}

In the previous section, we analyzed the FSO channel by assuming a given misalignment displacement $s$ and modeled the resulting pointing error accordingly.
However, in practice, determining a realistic value for $s$ is nontrivial, as it depends on various satellite dynamics and control imperfections. To make the analysis more practical, this section focuses on quantifying the misalignment component, considering a scenario where the transmitting satellite directs its optical signal toward the receiver satellite based on the their initial positional information. The given scenario is illustrated in Fig.~\ref{fig:displacement}. This allows us to characterize the worst beam displacement. 

Let the initial positions of the transmitting and receiving satellites be represented in spherical coordinates as
\begin{align}
\tilde{\bp}_{\mathrm{tx},0} &= (r_\mathrm{tx},\,\tilde{\theta}_\mathrm{tx},\,\tilde{\psi}_\mathrm{tx}), \\
\tilde{\bp}_{\mathrm{rx},0} &= (r_\mathrm{rx},\,\tilde{\theta}_\mathrm{rx},\,\tilde{\psi}_\mathrm{rx}),
\end{align}
where \( r \) denotes the radial distance from the center of Earth, and \( (\tilde{\theta}, \tilde{\psi}) \) are the polar and azimuthal angles measured in the Earth-centered inertial (ECI) coordinate frame.

Then, in our scenario of interest, the {optical} signal moves {from the initial transmitter position  $\tilde{\bp}_{\mathrm{tx},0}$ to the initial receiver satellite position $\tilde{\bp}_{\mathrm{rx},0}$.}
However, since the optical signal travels at finite speed, the receiver satellite continues to move along its orbit during the signal propagation time. 
Thus, at {arbitrary} time $\Delta t$ after beam transmission, the receiver's position becomes
\begin{equation}
\tilde{\bp}_{\mathrm{rx},\Delta t} = \left(r_\mathrm{rx},\,\tilde{\theta}_\mathrm{rx} + \dot{{\theta}}_{\mathrm{rx}} \Delta t,\,\tilde{\psi}_\mathrm{rx} + \dot{{\psi}}_{\mathrm{rx}} \Delta t\right),
\end{equation}
where $\dot{{\theta}}_{\mathrm{rx}}$ and $\dot{{\psi}}_{\mathrm{rx}}$ denote the polar and azimuthal angular velocities of receiver satellite in the ECI coordinate frame. Note that similar motion-induced displacement arises in terrestrial wireless networks. However, due to relatively low mobility and short transmission distances on the ground, the resulting impact is typically negligible. In contrast, high orbital velocities and long propagation distances in inter-satellite communication lead to more significant displacement during signal travel, making it a crucial factor in system modeling and performance analysis.

\begin{figure*}[b!]		
\hrule 
        \begin{equation}
			\label{eq: F2}
			F(\tau) = r_\mathrm{rx}^2+ r_\mathrm{tx}^2+
			2r_\mathrm{rx}r_\mathrm{tx} \left ( \sin (\theta_\mathrm{rx} + \omega_\mathrm{rx}\tau) \sin(\theta_\mathrm{tx})  \cos(\psi_\mathrm{rx}-\psi_\mathrm{tx}) +  \cos (\theta_\mathrm{rx} + \omega_\mathrm{rx}\tau)\cos(\theta_\mathrm{tx}) \right) - c^2 \tau^2. \tag{29}
		\end{equation}    
\end{figure*}

For notational simplicity, we rotate the coordinate system to a receiver-centric spherical coordinate frame. In this rotated frame, the angular motion of the receiver is consolidated into a single angular component, which is expressed as
\begin{equation}
\bp_{\mathrm{rx},\Delta t} = \left(r_\mathrm{rx},\,\theta_\mathrm{rx} + \omega_{\mathrm{rx}}\Delta t,\,\psi_\mathrm{rx}\right),
\end{equation}
where $(\theta, \psi)$ are the vertical and horizontal angles measured in the receiver-centric spherical coordinate frame, and $\omega_{\mathrm{rx}} = \sqrt{GM_\mathrm{E}/r_\mathrm{rx}^3}$ is the angular velocity determined by the satellite's altitude. The constants $G$ and $M_\mathrm{E}$ denote the gravitational constant and Earth's mass, respectively.
Note that the transmit satellite position also can be expressed as $\bp_{\mathrm{tx,0}}$ using the receiver-centric spherical reference frame.

\textbf{\textit{Remark:}} 
Although knowing the absolute coordinates of satellites in the ECI coordinate frame is critical for designing  entire satellite constellations and operating terrestrial-space communication networks, we adopt the receiver satellite-centered frame for this displacement computation. 
Considering the receiver satellite-centric coordinate system can be justified because the displacement solely depends on the movement of the receiver satellite, and the resulting displacement is independent of the selected reference frame.

To compute the propagation time $\tau$ required for the transmitted optical signal to reach the receiver satellite, we model the optical signal as a spherical wavefront expanding at the speed of light $c$ from the initial position of the transmitter. Expressing all satellite positions in Cartesian coordinates for mathematical convenience,  the arrival time $\tau$ must satisfy the following implicit equation
\begin{equation} \label{eq:tau_solve}
    (x_{\mathrm{rx},\tau} - x_{\mathrm{tx},0})^2 + (y_{\mathrm{rx},\tau} - y_{\mathrm{tx},0})^2 + (z_{\mathrm{rx},\tau} - z_{\mathrm{tx},0})^2 = c^2 \tau^2,
\end{equation}
which ensures that the receiver lies on the spherical wavefront at time $\tau$.
For mathematical convenience, all satellite positions are expressed in Cartesian coordinates. The position of satellite $i \in \{\mathrm{tx}, \mathrm{rx}\}$ at arbitrary time $\Delta t$ is given by
\begin{align}
    x_{i,\Delta t} &= r_{i} \sin(\theta_{i,\Delta t}) \cos(\psi_{i,\Delta t}), \label{eq:x} \\
    y_{i,\Delta t} &= r_{i} \sin(\theta_{i,\Delta t}) \sin(\psi_{i,\Delta t}), \label{eq:y}\\
    z_{i,\Delta t} &= r_{i} \cos(\theta_{i,\Delta t}).\label{eq:z}
\end{align}

\begin{algorithm}[t]
\caption{Bisection Method for Solving $F(\tau) = 0$}
\label{alg:bisection}
\begin{algorithmic}[1]
    \State \textbf{Input:} $F(\tau)$, parameters $r_{\mathrm{rx}}, r_{\mathrm{tx}}$, tolerance $\varepsilon$
    \State Set lower bound $\tau_+ \gets 0$
    \State Set upper bound $\tau_- \gets \sqrt{r_{\mathrm{rx}}^2 + r_{\mathrm{tx}}^2 + 4r_{\mathrm{rx}}r_{\mathrm{tx}}} \big/ c$
    \While{$\tau_- - \tau_+ > \varepsilon$}
        \State $\tau_m \gets (\tau_+ + \tau_-)/2$
        \If{$F(\tau_m) \cdot F(\tau_+) < 0$}
            \State $\tau_- \gets \tau_m$
        \Else
            \State $\tau_+ \gets \tau_m$
        \EndIf
    \EndWhile
    \State \textbf{Output:} Estimated solution $\tau \gets (\tau_+ + \tau_-)/2$
\end{algorithmic}
\end{algorithm}

\noindent {Using \eqref{eq:x}-\eqref{eq:z}, we can expand the implicit function in \eqref{eq:tau_solve} as in \eqref{eq: F2} at the bottom of this page and} obtain the numerical value of $\tau$ by finding certain instant that satisfies $F(\tau) = 0$. 
However, since $F(\tau)$ involves both trigonometric and polynomial terms, finding a closed-form solution for $\tau$ is generally intractable.
Instead, we first show that a feasible solution exists within a finite domain in the following theorem.

\begin{theorem} \label{thm: feasible solution}
A feasible solution \(\tau^*\) satisfying \(F(\tau^*) = 0\) always exists over the entire non-negative time domain \([0,\infty)\).
\end{theorem}

\begin{proof}
The function \(F(\tau)\) is continuous over the interval \([0,\infty)\), as it consists of trigonometric and polynomial terms. Thus, by the intermediate value theorem, a solution exists if two points in the interval yield opposite signs of \(F(\tau)\).

Let us consider \(\tau_{+} = 0\). At this point, we have
\setcounter{equation}{29}
\begin{align}
F(0) = 
2r_\mathrm{rx}r_\mathrm{tx} \bigl( &\sin(\theta_\mathrm{rx}) \sin(\theta_\mathrm{tx}) \cos(\psi_\mathrm{rx} - \psi_\mathrm{tx}) \notag \\
&+ \cos(\theta_\mathrm{rx}) \cos(\theta_\mathrm{tx}) \bigr) + r_\mathrm{rx}^2 + r_\mathrm{tx}^2.
\end{align}
Given that the elevation angles satisfy \(\theta_\mathrm{rx}, \theta_\mathrm{tx} \in [0,\pi]\), we know that 
{\(\sin(\theta_\mathrm{rx})\ge0 \) and \( \sin(\theta_\mathrm{tx}) \ge 0\)}. Hence, the term involving the cosine difference is bounded as
\begin{align}
- \sin(\theta_\mathrm{rx}) \sin(\theta_\mathrm{tx}) &\le \sin(\theta_\mathrm{rx}) \sin(\theta_\mathrm{tx}) \cos(\psi_\mathrm{rx} - \psi_\mathrm{tx}) \notag \\
&\le \sin(\theta_\mathrm{rx}) \sin(\theta_\mathrm{tx}).
\end{align}
From the result, we obtain the bounds for \(F(0)\) as
\begin{equation}
(r_{\mathrm{rx}} - r_{\mathrm{tx}})^2 \le F(0) \le (r_{\mathrm{rx}} + r_{\mathrm{tx}})^2,
\end{equation}
resulting in \(F(0) \ge 0\).

Next, we analyze the behavior of \(F(\tau)\) as \(\tau\) increases. Since the angular terms are bounded and their total contribution does not exceed 2, we can derive a conservative upper bound
\begin{equation}
F(\tau) \le r_{\mathrm{rx}}^2 + r_{\mathrm{tx}}^2 + 4r_{\mathrm{rx}}r_{\mathrm{tx}} - c^2\tau^2.
\end{equation}
Let us define \(\tau_{-} = \sqrt{r_{\mathrm{rx}}^2 + r_{\mathrm{tx}}^2 + 4r_{\mathrm{rx}}r_{\mathrm{tx}}}/c\). At this point, the inequality implies \(F(\tau_{-}) < 0\), regardless of the angular variables. Hence, it is satisfied that {\(F(\tau_{-}) <0\).}
From the fact that \(F(\tau)\) is continuous and satisfies \(F(\tau_{+}) \ge 0\), \(F(\tau_{-}) < 0\), we know there must exist some \(\tau^* \in [\tau_{+}, \tau_{-}]\) such that \(F(\tau^*) = 0\) by the intermediate value theorem. This concludes the proof.
\end{proof}
\noindent Although it is difficult to find a closed-form solution of $F(\tau) =0$ in general, Theorem \ref{thm: feasible solution} guarantees the existence of the solution $\tau^*$. Therefore, we approximate the solution numerically using the bisection method.
The entire procedure is summarized in Algorithm~\ref{alg:bisection}.

With the satellite position information at time $\tau$, we can now derive the displacement between the beam’s intended direction and the actual receiver location. Let $\ba = \bp_{\mathrm{rx},0} - \bp_{\mathrm{tx},0}$ denote the directional vector from the transmitter to the initially targeted receiver position, and let $\bb = \bp_{\mathrm{rx},\tau} - \bp_{\mathrm{tx},0}$ be the directional vector from the transmitter to the actual receiver position at time $\tau$. These vectors and the associated geometry are illustrated in Fig.~\ref{fig:displacement}.
Then, the displacement vector that can quantify the misalignment is obtained as 
\begin{equation}
\bs = \frac{\bb^\mathrm{T} \ba}{\Vert \ba \Vert^2}   \ba  - \bb.
\end{equation}
Here, the magnitude of this vector, i.e., $s = \Vert \bs \Vert$, represents the misalignment displacement induced by the relative motion of the receiver.

\begin{table}[t]
\small
\captionsetup{justification=centering, labelsep=newline, font={smaller,sc}}
\caption{Optical Transceiver and Link Parameters \cite{song2017impact, osborn2021adaptive}}
\label{tab:transceiver}
\centering
\begin{tabular}{lcl}
\hline
\textbf{Parameter} & \textbf{Symbol} & \textbf{Value} \\
\hline
Wavelength & $\lambda$ & $1550\,\mathrm{nm}$ \\
Receiver responsivity & $R$ & $0.87\,\mathrm{A/W}$ \\
Receiver aperture radius& $a$ & $20\,\mathrm{cm}$ \\
Jitter angular deviation & $\sigma_j$ & $8 \times10^{-6}\,\mathrm{rad}$ \\
Beam waist radius & $w_0$ & $1.25\times10^{-2} \,\mathrm{m}$ \\
Optical noise variance & $\sigma_n^2$ & $1.6\times10^{-14}\,\mathrm{A}^2$\\
Transmit power & $P_t$ & 28\,dBm \\
Target transmission rate & $R_0$ & $1\,\mathrm{Gbps}$ \\
\hline
\end{tabular}
\end{table}

\section{Simulation Results} \label{sec5}
In this section, we investigate the pointing error induced by joint jitter and beam misalignment in inter-satellite FSO communication systems, which is denoted by \textit{Beam misalignment}. For a baseline, we consider a perfect beam alignment scenario, denoted by  \textit{No misalignment}, where the FSO channel is only affected by jitter-induced pointing error. To accurately describe practical ISL configurations, we consider two representative Walker constellation models, Walker-Star for Iridium and  Walker-Delta for Starlink.
The Iridium constellation consists of 66 satellites deployed across 6 OPs at an altitude of $781\,\mathrm{km}$, with 11 satellites per plane. The average intra-plane and inter-plane inter-satellite spacings are approximately $4{,}085\,\mathrm{km}$ and $3{,}745\,\mathrm{km}$, respectively~\cite{anderson:2020iridium}.
In contrast, the Starlink Phase I constellation comprises 1,584 satellites distributed over 72 OPs at an altitude of $550\,\mathrm{km}$, with 22 satellites per plane. The corresponding intra-plane and inter-plane spacings are approximately $1{,}977\,\mathrm{km}$ and $604\,\mathrm{km}$, respectively~\cite{liang:2021phasing}.
Based on these constellation configurations, we simulate inter-satellite FSO links using a fixed set of transceiver parameters. Unless otherwise specified, the simulation parameters for the satellite optical transceiver follow the values in Table~\ref{tab:transceiver}.

\begin{figure}[t!]
    \centering
    \includegraphics[width=1\linewidth]{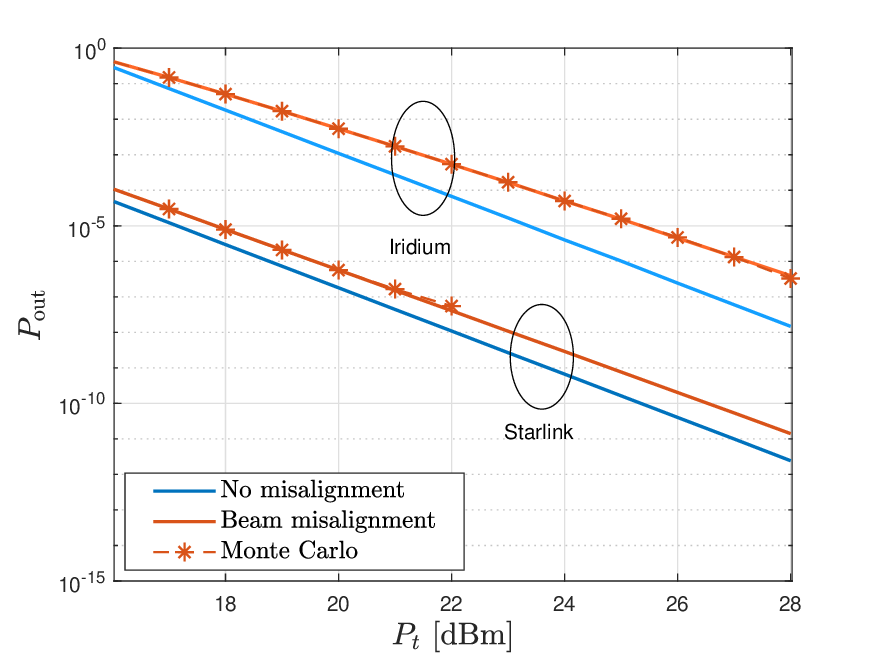}
    \caption{Outage probability according to $P_t$ in an intra-OP link.}
    \label{fig:intraOP}
\end{figure}
\begin{figure}[t!]
    \centering
\includegraphics[width=1\linewidth]{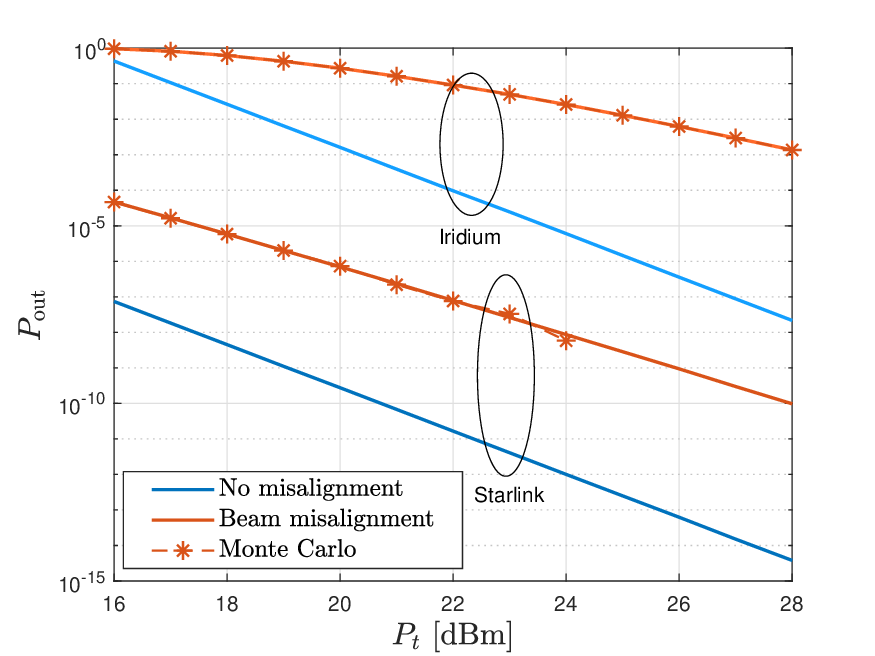}
    \caption{Outage probability according to $P_t$ in an inter-OP link.}
    \label{fig:interOP}
\end{figure}

Fig.~\ref{fig:intraOP} demonstrates the outage probability for FSO communication within the intra-OP at a target transmission rate of \(R_0 = 1 \,\mathrm{Gbps}\), comparing the Iridium and Starlink constellations. It is observed that beam misalignment results in noticeable performance degradation in both constellations, as reflected by the increased outage probability. Despite this degradation, the impact of misalignment can be mitigated by allocating additional transmit power. For instance, achieving an outage probability of \(10^{-8}\) requires a transmit power of \(P_t =  22\,\mathrm{dBm}\) with no misalignment case, whereas misalignment increases this requirement to \(P_t =  23\,\mathrm{dBm}\) in the Starlink case. When comparing Iridium and Starlink, the performance gap between with and without misalignment is more significant in Iridium. 
This can be explained by the configuration of Starlink, which places more satellites at lower orbital altitudes, resulting in less relative movement and therefore smaller misalignment.

To verify the effectiveness of the proposed algorithm, we also compare the truncated version of the CDF with results obtained from Monte Carlo simulations. Since the PDF in \eqref{eq:pdf} lacks a closed-form inverse, we employ rejection sampling for the Monte Carlo simulations, which is suitable for arbitrary distributions \cite{robert2004monte}. The results show that the outage probability computed using the proposed algorithm (solid line) closely matches that obtained from the Monte Carlo simulations (dashed line). The strong agreement between the analytical and simulation-based results validates the accuracy of the proposed approach, and also demonstrates that the truncated CDF expression derived in our algorithm can effectively capture the statistical behavior of the system without relying on exhaustive sampling.

In Fig.~\ref{fig:interOP}, we analyze the outage probability performance for inter-OP links in both the Iridium and Starlink constellations. Compared to intra-OP links, the impact of misalignment becomes more severe in inter-OP links for both constellations. This performance degradation results from the relative motion in inter-OP links, where the receiving satellite moves in a direction different from the beam propagation. Although the reduced ISL distance improves performance under no misalignment case, the advantage is diminished by the effect of increased misalignment. Therefore, robust FSO transceiver design is more critical for maintaining reliable communication in inter-OP links, even in densely populated mega-constellations such as Starlink.

\begin{figure}
    \centering
    \includegraphics[width=1\linewidth]{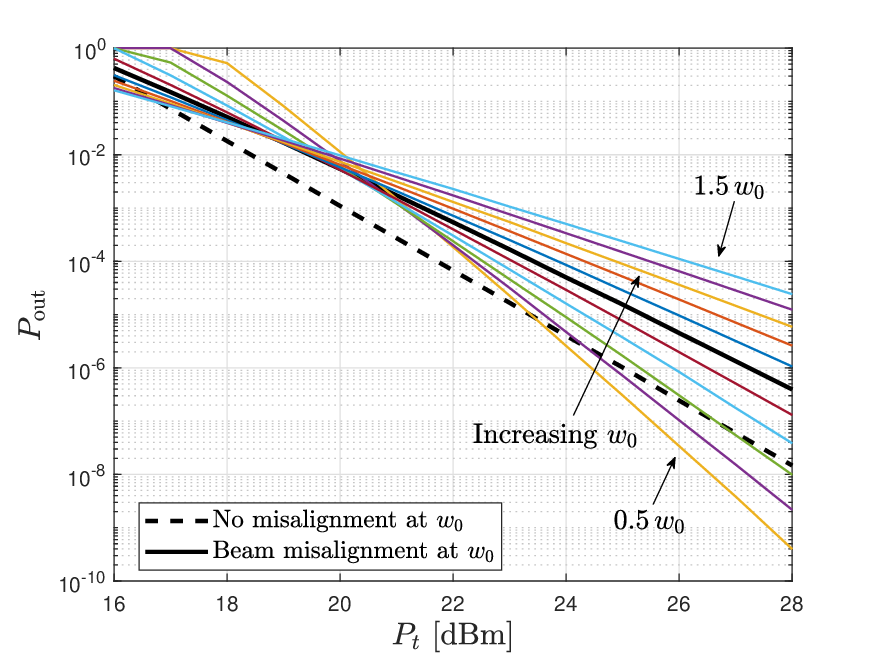}
    \caption{Outage probability according to $w_0$ in an Iridium intra-OP link.}
    \label{fig:w0}
\end{figure}

Fig.~\ref{fig:w0} shows how the outage probability varies with the transmit beam waist radius, ranging from $0.5\,w_0$ to $1.5\,w_0$, where $w_0 = 1.25 \times 10^{-2}\, \mathrm{m}$. According to the beam propagation model in~\eqref{eq:beam radius}, a smaller beam waist leads to stronger diffraction, resulting in a larger beam radius as the propagation distance increases. This beam divergence directly affects the link performance depending on the available transmit power. When the transmit power is sufficiently high, a wider beam can maintain adequate power density across a broader beam footprint, thereby reducing the outage probability. In contrast, under limited power conditions, distributing the beam over a wide area lowers the power density throughout the beam footprint, increasing the likelihood of outage. In such cases, confining the beam to a smaller area helps maintain sufficient received power. This observation is consistent with previous findings in~\cite{bloom2003understanding}, which emphasize the importance of controlling beam divergence to ensure reliable optical links under misalignment or platform motion.
An additional observation from the simulation results is that, with an appropriately chosen beam waist, better outage performance can be achieved under \textit{Beam misalignment} (solid line) than under \textit{No misalignment} with a fixed beam waist (dotted line). While using a smaller beam waist under \textit{No misalignment} also results in improved performance at higher transmit power levels, this observation suggests that transmitter-side beam waist optimization can mitigate the performance degradation caused by misalignment, potentially reducing the reliance on highly precise tracking systems.

\begin{figure}
    \centering
\includegraphics[width=1\linewidth]{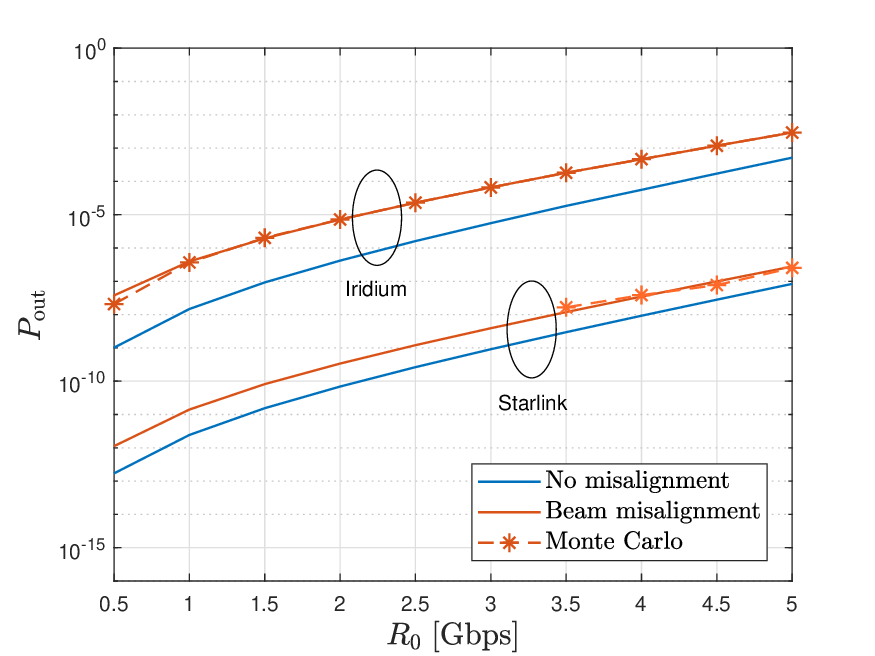}
    \caption{Outage probability according to $R_0$ in an intra-OP link.}
    \label{fig:R0_intra}
\end{figure}

\begin{figure}
    \centering
\includegraphics[width=1\linewidth]{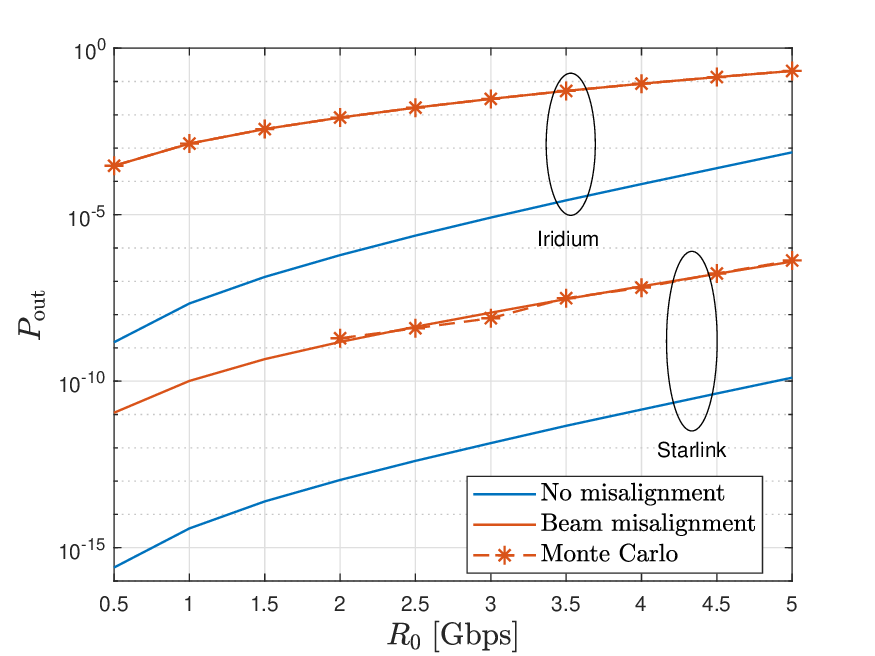}
    \caption{Outage probability according to $R_0$ in an inter-OP link.}
    \label{fig:R0_inter}
\end{figure}

Figs.~\ref{fig:R0_intra} and~\ref{fig:R0_inter} present the outage probability as a function of the target transmission rate $R_0$. An increase in the target rate leads to a higher required received power, which results in an elevated outage probability. As the rate threshold increases, the difference in performance between with and without misalignment becomes smaller. The diminishing gap can be attributed to the reduced impact of the tail probabilities in both Rayleigh and Rician fading distributions at high thresholds. The Starlink constellation, characterized by its dense orbital structure and shorter inter-satellite link distances, maintains relatively low outage levels even under high target rate requirements. In contrast, the Iridium system shows significant performance degradation as the target rate increases, yielding much higher outage probabilities. The inter-OP scenario demonstrates a greater performance disparity between with and without misalignment cases, consistent with previously observed results in Figs.~\ref{fig:intraOP} and~\ref{fig:interOP}.

\begin{figure}
    \centering
\includegraphics[width=1\linewidth]{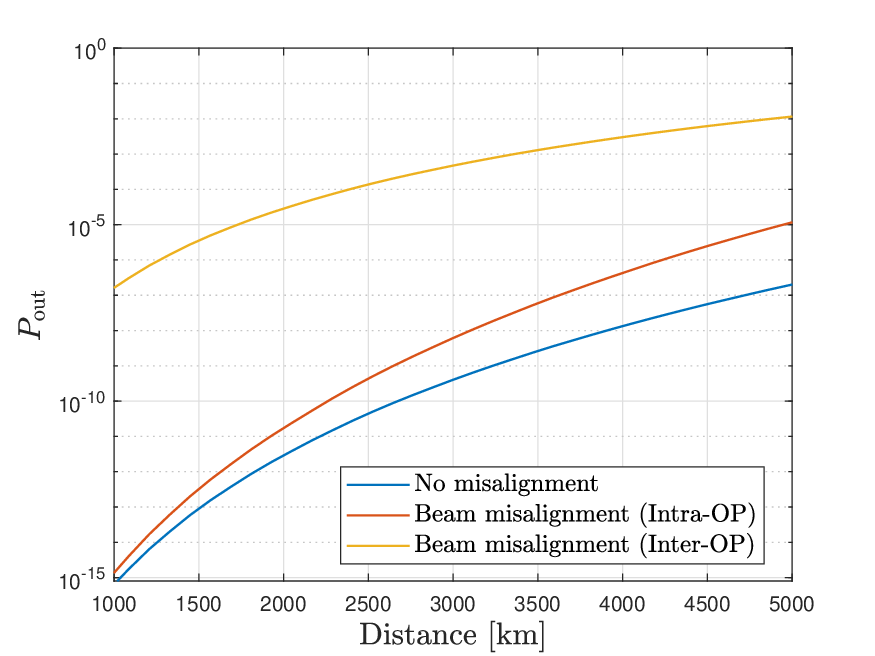}
    \caption{Outage probability according to ISL distance between Tx satellite and Rx satellite with $P_\mathrm{t} = 28 \,\mathrm{dBm}$ and $R_0 = 1\, \mathrm{Gbps}$.}
    \label{fig:distance}
\end{figure}

In Fig.~\ref{fig:distance}, the outage probability performance is analyzed with respect to the ISL distance. The analysis considers only one transmit satellite and one receive satellite at an altitude of $550\,\mathrm{km}$, ignoring the full constellation. In the case of inter-OP links, the satellites are assumed to be located on co-phased orbital planes, meaning that both satellites maintain phase alignment along their respective orbits.
Although previous analysis of inter-satellite FSO links in~\cite{song2017impact} has demonstrated link distance-dependent results, they ignored the relative motions of the satellites. However, under practical conditions where such ideal alignment cannot be guaranteed, the orbital configuration and relative motion of the satellites introduce additional performance variations. Differences between the pointing direction of the transmitted beam and the trajectory of the receiving satellite become more significant in the inter-OP configuration. For intra-OP links, the directional mismatch remains minimal, resulting in a relatively small deviation from the ideal alignment case. In contrast, inter-OP links exhibit greater displacement between the beam center and the receiver due to rapid cross-plane motion, leading to a significantly increased outage probability.

Fig.~\ref{fig:bar_graph} presents the number of satellites and OPs required to meet a given outage probability under beam misalignment conditions. As the target outage probability becomes more stringent, a denser constellation is needed to ensure sufficient spatial coverage and link reliability. Notably, a decrease in the acceptable outage probability leads to a significant increase in the number of orbital planes, while the number of satellites per plane remains relatively stable. Based on the results, it becomes possible to estimate the total number of satellites required for a given performance requirement, offering insight into the rationale behind deploying a large number of OPs, as seen in the 72-plane configuration of Starlink.
The simulation results show that a large satellite constellation with numerous satellites and OPs can provide a robust communication framework, with the impact of beam misalignment being effectively minimized. Thus, constellation design can play an important role in ensuring the reliability of the communication system.

\begin{figure}
    \centering
\includegraphics[width=1\linewidth]{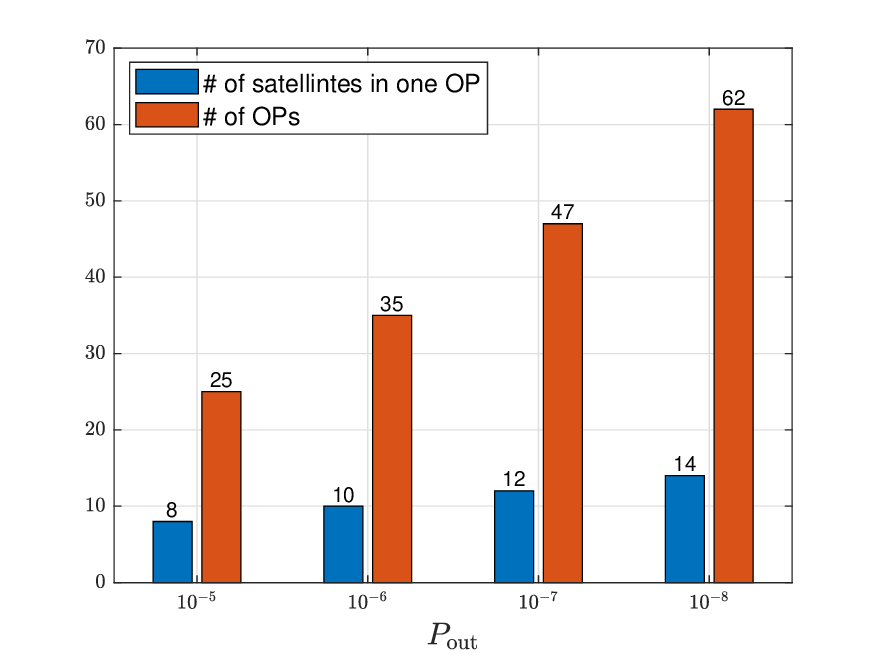}
    \caption{Required numbers of satellites and OPs to achieve given outage probability level at an altitude of $550\, \mathrm{km}$.}
    \label{fig:bar_graph}
\end{figure}

\section{Conclusion} \label{sec6}
In this paper, we analyzed the channel characteristics of inter-satellite FSO communication under the joint jitter and misalignment-induced pointing error. We derived a closed-form expression for the CDF of inter-satellite FSO channel and proposed a truncated CDF formulation to efficiently compute outage probabilities with low complexity. We introduced a bisection algorithm to efficiently determine the optimal truncation index. To reflect the impact of satellite dynamics, we quantified the displacement caused by relative motion under a worst-case scenario, which served as a lower bound for performance evaluation. Numerical results verified that the proposed truncated CDF closely matched with the Monte Carlo simulations and showed that beam misalignment leads to noticeable performance degradation compared to no misalignment case, especially in inter-OP links. The analysis presented in this study provides valuable guidance for designing satellite networks, helping to determine the required hardware performance levels, as well as the numbers of satellites and OPs for a satellite constellation design.

\section*{Appendix A\\ Proof of Lemma~I}

To derive the PDF of the FSO channel gain \( h \), a transformation is applied to the Rician-distributed radial displacement \( r \) given in \eqref{eq:r}. The transformation is defined by the channel gain function \( h = g(r) \), and its inverse is given by
\begin{equation}
    g^{-1}(h) = \sqrt{ -\frac{\omega_{\ell_{\rm eq}}^2}{2} \ln\left( \frac{h}{A_0} \right) }, \quad 0 \le h \le A_0,
\end{equation}
which maps a channel gain value to the corresponding radial displacement.
For convenience, we define
\begin{equation}
    k = \frac{\omega_{\ell_{\rm eq}}^2}{2} \ln\left( \frac{h}{A_0} \right),
\end{equation}
so that \( g^{-1}(h) = \sqrt{-k} \), and the transformation Jacobian becomes
\begin{equation}
    \left| \frac{d g^{-1}(h)}{dh} \right|
    = \left( \frac{\omega_{\ell_{\rm eq}}^2}{2h} \right)
       \frac{1}{2 \sqrt{-k} }.
\end{equation}
Using the change-of-variable technique for PDFs, the PDF of the FSO channel gain can be written as
\begin{align}
    f_h(h)
    &= f_r(g^{-1}(h)) \cdot \left| \frac{d g^{-1}(h)}{dh} \right| \notag \\
    &= \frac{ \sqrt{-k} }
             { \ell^2 \sigma_j^2 }
       e^{-\frac{-k + s^2}{2 \ell^2 \sigma_j^2}}
     I_0\left(
           \frac{s}{\ell^2 \sigma_j^2}\sqrt{-k}
       \right)
       \cdot \left( \frac{ \omega_{\ell_{\rm eq}}^2 }{ 2h } \right)
       \frac{1}{2 \sqrt{-k} }.
\end{align}
Simplifying the above expression, the terms involving \( \sqrt{-k} \) cancel, yielding
\begin{align}
    f_h(h)
    &= \frac{ \omega_{\ell_{\rm eq}}^2 }{ 4 \ell^2 \sigma_j^2 }
        e^{-\frac{s^2}{2 \ell^2 \sigma_j^2}}
        \frac{1}{h}
       \left( \frac{h}{A_0} \right)^{ \frac{ \omega_{\ell_{\rm eq}}^2 }{ 4 \ell^2 \sigma_j^2 } }
        \notag \\
    &\quad \cdot I_0\left(
           \frac{s}{\ell^2 \sigma_j^2}
            \sqrt{ -\frac{ \omega_{\ell_{\rm eq}}^2 }{2} \ln \left( \frac{h}{A_0} \right) }
       \right).
\end{align}
To further simplify the expression, introduce a dimensionless parameter
\begin{equation} \label{eq:gamma}
    \gamma = \frac{ \omega_{\ell_{\rm eq}} }{ 2 \ell \sigma_j },
\end{equation}
so that the PDF becomes
\begin{equation}
    f_h (h ) = \frac{\gamma^2}{A_0^{\gamma^2}}  e^{- \frac{s^2}{2\ell^2\sigma_j^2} } h^{\gamma^2 -1} I_0\left(\frac{s}{\ell^2\sigma_{j}^2} \sqrt{-\frac{\omega_{\ell_{\rm eq}}^2}{2} \ln{\frac{h}{A_0}}} \right),
\end{equation}
which provides the closed-form expression for the PDF of the FSO channel gain under jointly distributed pointing error and misalignment.

\section*{Appendix B\\ Proof of Lemma~2}

The CDF of the FSO channel gain \( h \) can be obtained by integrating its probability density function \( f_h(h) \), given in~\eqref{eq:pdf}, as follows
\begin{equation}
    F_h(x) = \int_0^x f_h(t)\,dt.
\end{equation}
Substituting the expression of the PDF yields
\begin{equation} \label{eq:cdf_initial}
    F_h(x) = A_0' \int_0^x t^{\gamma^2 -1} I_0\left( \frac{s}{\ell^2 \sigma_j^2} \sqrt{-\frac{\omega_{\ell_{\rm eq}}^2}{2} \ln \frac{t}{A_0}} \right) dt,
\end{equation}
where \( A_0' = \frac{\gamma^2}{A_0^{\gamma^2}} e^{ -\frac{s^2}{2 \ell^2 \sigma_j^2} } \), and \( I_0(\cdot) \) denotes the modified Bessel function of the first kind and order zero.

To enable analytical tractability, we apply the power series expansion of the Bessel function
\begin{equation}
    I_0(p) = \sum_{n=0}^\infty \frac{(p/2)^{2n}}{(n!)^2}.
\end{equation}
Applying this expansion to~\eqref{eq:cdf_initial} leads to
\begin{align}
    &F_h(x) \notag \\&= A_0' \int_0^x t^{\gamma^2 -1} \sum_{n=0}^{\infty} \frac{1}{(n!)^2} \left(  \frac{s}{2\ell^2 \sigma_j^2} \sqrt{ \frac{-\omega_{\ell_{\rm eq}}^2}{2} \ln \frac{t}{A_0} } \right)^{2n} dt \notag \\
    &= A_0' \sum_{n=0}^{\infty} \frac{1}{(n!)^2} \left( \frac{s^2 \omega_{\ell_{\rm eq}}^2}{8 \ell^4 \sigma_j^4} \right)^n \int_0^x t^{\gamma^2 -1} \left( \ln \frac{A_0}{t} \right)^n dt.
\end{align}
To simplify notation, we define the constant
\begin{equation} \label{eq:S_n}
    S_n = \frac{s^2 \omega_{\ell_{\rm eq}}^2}{8 \ell^4 \sigma_j^4},
\end{equation}
so that the expression becomes
\begin{equation}
    F_h(x) = A_0' \sum_{n=0}^{\infty} \frac{S_n^n}{(n!)^2} \int_0^x t^{\gamma^2 - 1} \left( \ln \frac{A_0}{t} \right)^n dt.
\end{equation}
Next, perform the change of variables \( u = \ln\left(A_0 / t\right) \), from which it follows that \( t = A_0 e^{-u} \) and \( dt = -A_0 e^{-u} du \). Then, the integral becomes
\begin{align}
    F_h(x) &= A_0' A_0^{\gamma^2} \sum_{n=0}^{\infty} \frac{S_n^n}{(n!)^2} \int_{\ln(A_0/x)}^\infty e^{-\gamma^2 u} u^n du.
\end{align}
By substituting \( \xi = \gamma^2 u \) and \( d\xi = \gamma^2 du \), we rewrite the expression as
\begin{align}
    F_h(x) = A_0' A_0^{\gamma^2} \sum_{n=0}^{\infty} \frac{S_n^n}{(n!)^2 \gamma^{2(n+1)}} \int_{\gamma^2 \ln(A_0/x)}^\infty e^{-\xi} \xi^n d\xi.
\end{align}
Finally, by using the previously defined parameter \( \gamma \) in~\eqref{eq:gamma}, and the definition of \( S_n \) in \eqref{eq:S_n}, the CDF expression simplifies to
\begin{equation}
    F_h(x) = e^{- \frac{s^2}{2 \ell^2 \sigma_j^2}} \sum_{n=0}^{\infty} \frac{ \left( \frac{s^2}{2 \ell^2 \sigma_j^2} \right)^n }{(n!)^2} \Gamma\left(n + 1, \gamma^2 \ln \frac{A_0}{x} \right),
\end{equation}
where \( \Gamma(\cdot\,, \cdot) \) denotes the upper incomplete Gamma function, defined as
\begin{equation}
    \Gamma(a, b) = \int_b^\infty t^{a-1} e^{-t} dt.
\end{equation}
This completes the derivation of the CDF of the FSO channel gain.

\bibliographystyle{IEEEtran}
\bibliography{FSO_references}

\end{document}